\documentclass[11pt, letterpaper]{article}
\usepackage[utf8]{inputenc}
\usepackage{amsmath}
\usepackage{amsfonts}
\usepackage{amsthm}
\usepackage{amssymb}
\usepackage{fullpage}
\usepackage[dvipsnames]{xcolor}
\usepackage{xspace}
\usepackage{graphicx}
\usepackage[colorlinks=true,allcolors=MidnightBlue]{hyperref}
\usepackage{cleveref}
\usepackage{comment}
\usepackage{subcaption}
\usepackage{url}
\usepackage{nicefrac}
\usepackage{afterpage}
\usepackage{fullpage}
\usepackage{enumerate}
\usepackage{nicefrac}
\usepackage{lineno}
\usepackage[backend=bibtex,maxbibnames=99,sortcites]{biblatex}
\newtheorem{theorem}{Theorem}

\newtheorem{lemma}[theorem]{Lemma}
\newtheorem{claim}{Claim}
\newtheorem{cor}[theorem]{Corollary}
\newtheorem*{conjecture*}{Conjecture}

\newcommand{\R}{\ensuremath{\mathbb{R}}\xspace}
\newcommand{\ER}{\ensuremath{\exists\mathbb{R}}\xspace}
\newcommand{\N}{\ensuremath{\mathbb{N}}\xspace}
\newcommand{\A}{\ensuremath{\mathcal{A}}\xspace}
\newcommand{\NP}{\ensuremath{\mathrm{NP}}\xspace}
\newcommand{\PP}{\ensuremath{\mathrm{P}}\xspace}
\newcommand{\Tall}{\ensuremath{T^+}}
\newcommand{\Tint}{\ensuremath{T}}

\crefname{figure}{Figure}{Figures}
\crefname{theorem}{Theorem}{Theorems}
\crefname{lemma}{Lemma}{Lemmas}
\crefname{cor}{Corollary}{Corollaries}
\crefname{table}{Table}{Tables}
\crefname{section}{Section}{Sections}
\crefname{claim}{Claim}{Claims}

\newcommand{\GEM}[2]{{\ensuremath{\textsc{GEM}_{#1 \rightarrow #2}}}\xspace}
\newcommand{\PLEM}[2]{{\ensuremath{\textsc{Embed}_{#1 \rightarrow #2}}}\xspace}

\newcommand{\conv}{{\ensuremath{\textrm{conv}}}\xspace}

\newcommand{\stretcha}{{\ensuremath{\textsc{Stretchability}}}\xspace}
\newcommand{\seg}{{\ensuremath{\textrm{seg}}}\xspace}

\newcommand{\rafter}{rafter\xspace}
\newcommand{\htriangle}{helper triangle\xspace}

\newcommand{\sloop}{tunnel loop\xspace}
\renewcommand{\loop}{{\ensuremath{o}}\xspace}

\renewcommand\labelenumi{\roman{enumi})}
\renewcommand\theenumi\labelenumi
\title{Geometric Embeddability of Complexes is \ER-complete}
\author{
\parbox{4.5cm}{\center
{\textsc{Mikkel Abrahamsen}}\\[6pt]
\small
{miab@di.ku.dk}\\
{University of Copenhagen}}\and 
\parbox{4.5cm}{\center
{\textsc{Linda Kleist}}\\[6pt]
\small
{kleist@ibr.cs.tu-bs.de}\\
{TU Braunschweig}}\and 
\parbox{4.5cm}{\center
{\textsc{Tillmann Miltzow}}\\[6pt]
\small
{{t.miltzow@uu.nl}}\\
{Utrecht University}}\\}
\date{November 2021}

\begin{document}

\maketitle

\setcounter{page}{0}
\thispagestyle{empty}

\begin{abstract}
    We show that the decision problem of determining whether a given (abstract simplicial) $k$-complex has a geometric embedding in $\R^d$ is complete for the Existential Theory of the Reals for all $d\geq 3$ and $k\in\{d-1,d\}$.  This implies that the problem is polynomial time equivalent to determining whether a polynomial equation system has a real solution. Moreover, this implies NP-hardness and constitutes the first hardness result for the algorithmic problem of geometric embedding (abstract simplicial) complexes. 
\end{abstract}

\vfill

\begin{figure}[htb]
\centering
\includegraphics[scale=.65]{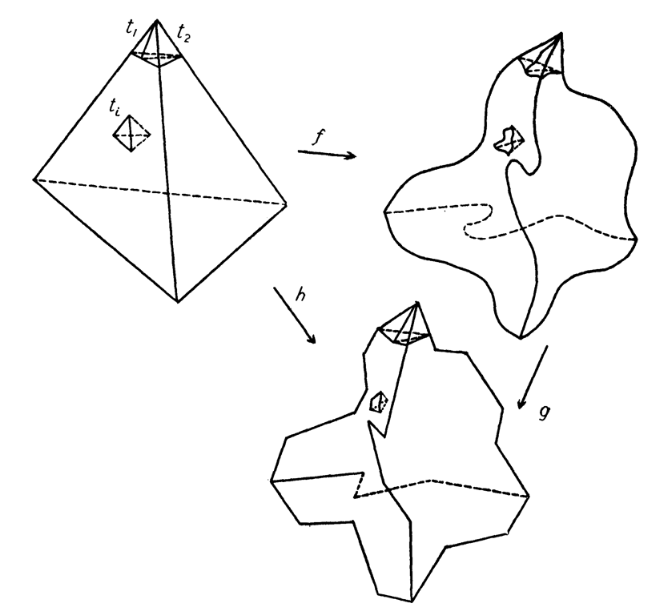}

\vspace{1.8cm}

\caption{Illustration of different embeddings of a complex; figure taken from Bing~\cite[Annals of Mathematics~$1959$]{bing59}.}
\label{fig:notions}
\end{figure}

\section{Introduction}
Since the dawn of the last century, much attention has been devoted to studying embeddings of complexes~\cite{gundert2018complexity, bing59, dieudonne2009history,grunbaum1969imbeddings,grunbaum1970,Menger-Dimension,shapiro1957obstructions,van1933komplexe}.
Typical types of embeddings include geometric (also referred to as linear), piecewise linear (PL), and topological embeddings, see also \cref{fig:notions}. For formal definitions, we refer to \cref{sec:definitions}; here we give an illustrative example.
Embeddings of a $1$-complex in the plane correspond to drawings of a graph in the plane. 
In a topological embedding, each edge is represented by a Jordan arc, in a PL embedding it is a concatenation of a finite number of segments, and in a geometric embedding each edge is represented by a segment. 
Unlike in higher dimensions,
for the embeddability of complexes in the plane all three notions coincide.

We are interested in the problem of deciding whether a given $k$-complex has a linear/piecewise linear/topological embedding in $\R^d$.
Several necessary and sufficient conditions are easy to identify and have been known for many decades. 
For instance, a $k$-simplex requires $k+1$ points in general position in $\R^d$ and, thus,  $k\leq d$ is an obvious necessary condition.
Moreover, it is straight-forward to verify that every set of $n$ points in $\R^3$ in general position allows for a geometric embedding of any 1-complex on $n$ vertices, i.e., the points are the vertices of a straight-line drawing of a (complete) graph.
Indeed, this fact generalizes to higher dimensions:
every $k$-complex embeds (even linearly) in~$\R^{2k+1}$~\cite{Menger-Dimension}. 
Van Kampen and Flores~\cite{flores1933n,skopenkov2014realizability,van1933komplexe} showed 
that this bound is tight by providing $k$-complexes that do not topologically embed into $\R^{2k}$.
For some time, it was believed that the existence of a topological embedding also implies the existence of a geometric embedding, e.g., 
Gr\"{u}nbaum conjectured that if a $k$-complex topologically embeds in~$\R^{2k}$, then it also geometrically embeds in~$\R^{2k}$~\cite{grunbaum1969imbeddings}. 
However, this was later disproven.
In particular, for every $k,d\geq 2$ with $k+1 \leq d \leq2k$, there exist $k$-complexes that have a PL embedding in $\R^d$, but no geometric embedding in $\R^d$~\cite{brehm-Linear-vs-PL, bokowski2000generation,brehm1983nonpolyhedral}. 
In contrast, PL and topological embeddability coincides in many cases, e.g.,  if $d\leq 3$~\cite{bing59,papakyriakopoulos1943} or $d-k \geq 3$ \cite{Bryant1972Approximating}.
There are many further  necessary and sufficient conditions known for geometric embeddings~\cite{alfonsin2005knots,novik2000note,JGAA-48,skopenkov2014realizability,timmreck2008necessary,Timmreck2015} and PL or/and topological embeddings~\cite{carmesin2019embedding,freedman1994,parsa2020,shapiro1957obstructions,ummel,skopenkov2016GraphProducts}.
\medskip

In recent years, the \textbf{algorithmic complexity} of deciding whether or not a given complex is embeddable gained attention.
In the absence of a complete characterization, an efficient algorithm is the best tool to decide embeddability.
For instance, deciding whether a  1-complex embeds in the plane corresponds to testing graph planarity and is thus polynomial time decidable~\cite{hopcroft1974efficient}. 
Conversely, the proven non-existence of efficient algorithms may offer a rigorous proof that a complete characterization is impossible. To give a concrete example, let \PLEM{k}{d} denote the algorithmic problem of determining whether a given
$k$-complex has a PL embedding in $\mathbb R^d$. 
Because \PLEM{4}{5} is known to be undecidable~\cite{undecidable2020,Skopenkov2020},
we have a proof that there does not exist an efficient algorithm for \PLEM{4}{5} -- even without any complexity assumptions (such as $\NP\neq \PP$ or similar).

More recently, there have been several breakthroughs concerning the \textbf{PL embeddability}.
For an overview of the state of the art, consider \cref{tab:resultsPL}.
\begin{table}[htb]
\centering
\caption{Overview of the complexity of \PLEM{k}{d}.}
\label{tab:resultsPL}
\includegraphics[page = 3]{figures/tables}
\end{table}
In dimensions $d\geq 4$, the decision problem \PLEM{k}{d}  is polynomial-time decidable for $k < \nicefrac{2}{3}\cdot (d-1)$ \cite{vcadek2017algorithmic,vcadek2014computing,cadek2014polynomial,krvcal2013polynomial} and NP-hard for all remaining cases~\cite{matousek2011hardness}, i.e., for all $k \geq \nicefrac{2}{3}\cdot (d-1)$.
For $d\geq 5$ and $k\in\{d-1,d\}$, \PLEM{k}{d} is even known to be undecidable; for all other NP-hard cases and $d\geq 4$ decidability is unknown.
For the case $d=3$,  Matoušek, Sedgwick, Tancer, and Wagner have shown that \PLEM{2}{3} and \PLEM{3}{3} are decidable \cite{matouvsek2018embeddability} and de Mesmay, Rieck, Sedgwick, and Tancer have proved NP-hardness~\cite{mesmay2020embeddability}.

Building upon \cite{matousek2011hardness}, 
Skopenkov and Tancer~\cite{skopenkov2019hardness} proved NP-hardness for a relaxed notion called \emph{almost (PL/topological) embeddability} where  
it is only required that
disjoint sets must be mapped to disjoint objects, i.e., this notions allows that two edges incident to a common vertex cross in an interior point.
More precisely, they showed 
that  recognizing almost embeddability of $k$-complexes in~$\mathbb R^d$ is \NP-hard for all  $d,k\geq 2$ such that
 $d\pmod 3 =1$ and $\nicefrac{2}{3}\cdot (d-1)\leq k\leq d$.
\medskip

The analogous questions for \textbf{geometric embeddings} are wide open. Let  \GEM{k}{d} denote the algorithmic problem of determining whether a given $k$-complex has a geometric embedding in $\mathbb R^d$. In contrast to PL embeddability, however, it is easy to see that $\GEM{k}{d}$ is decidable for all $k,d$, since every instance can be expressed as a sentence in the first order theory of the reals, which is decidable.
In analogy to the PL embeddings, \NP-hardness has been conjectured by Skopenkov.

\begin{conjecture*}[\cite{skopenkov2020invariants}, Conjecture 3.2.2]
    \GEM{k}{d} is \NP-hard for all $k,d$ with  \[\max \{3,k\} \leq d
\leq  \nicefrac{3}{2}\cdot k + 1.\]
\end{conjecture*}
Note that these parameters correspond to the NP-hard cases of \PLEM{k}{d}, see also \cref{tab:resultsPL}.
Cardinal~\cite[Section 4]{CardinalSurvey} mentions 
{$\GEM{2}{3}$} as an interesting open problem.
% in $\R^3$.
The closely related question of polyhedral complexes, posed in the Handbook of Discrete and Computational Geometry, reads as follows:
   When is a given finite poset isomorphic to the face poset of some polyhedral complex in a given space $\R^d$?~\cite[Problem 20.1.1]{Chapter-Polyhedral-Maps}.
Note that simplicial complexes are special 
cases of polyhedral complexes,
because each simplex is a basic polyhedron.
The  recognition of polyhedral complexes (with triangles and quadrangles) {in $\R^3$} has been claimed to be $\ER$-complete~\cite[Theorem~5]{CardinalSurvey}.
Focussing on convex polytopes, Richter-Gebert proved that recognizing convex polytopes in $\R^4$ is \ER-complete~\cite{richterGebertrealization,richter1995realization}.

\paragraph{Our Results.}
In this work, we present the first results concerning Skopenkov's conjecture for any non-trivial entry with $d\geq 3$. 
More precisely, we establish the exact computational complexity  of \GEM{k}{d} for all values $d\geq 3$ and  $k \in \{ d-1,d\}$, hereby confirming the conjecture for these cases.
This includes a complete understanding of the most intriguing entries with~$d=3$. 
Note that this also answers the {computational aspects of the} question from the Handbook of Discrete and Computational Geometry.
\cref{tab:results} summarizes the current knowledge on the computational complexity of~\GEM{k}{d}.

\begin{theorem}
\label{thm:main}
For every $d\geq 3$ and each $k \in \{ d-1,d\}$,  the decision problem \GEM{k}{d} is \ER-complete.
Moreover, the statement remains true even if a PL embedding is given.
\end{theorem}

\begin{table}[htb]
\centering
\caption{Overview of the computational complexity of \GEM{k}{d}.}
\label{tab:results}
\includegraphics[page = 5]{figures/tables}
\end{table}

Our proof implies that distinguishing between $k$-complexes with PL and geometric embeddings in $\mathbb R^d$ is complete for $\ER$.
Because $\NP\subseteq\ER$, our result confirms the conjecture by Skopenkov for the corresponding values of $k$ and $d$. 
Moreover, if $\NP\neq\ER$, the problem \GEM{k}{d} cannot be tackled with well developed tools for \NP-complete problems such as SAT and ILP solvers.
For more details, we refer to \cref{sec:ExReals}.
\medskip

A geometric embedding of a complex can also be viewed as a \emph{simplicial representation} of a hypergraph, i.e., a representation of a hypergraph in which every hyperedge is represented by a simplex. Of particular interest is the case of uniform hypergraphs where all hyperedges have the same number of elements. 
Thus, in the language of hypergraphs, our result reads as follows.
\begin{cor}
For all $d\geq 3$ and every $k \in \{ d-1,d\}$, deciding whether a $(k+1)$-uniform hypergraph has a simplicial representation in $\R^d$ is  \ER-complete. 
\end{cor}

\paragraph{Outline and techniques.}
Our proof of \cref{thm:main} consists of three steps:  Establishing \ER-membership,
showing \ER-hardness in~$\R^3$, i.e., of \GEM{2}{3} and \GEM{3}{3},
and reducing \GEM{k}{d} to \GEM{k+1}{d+1}. The core of the proof lies in establishing hardness of $\GEM{2}{3}$.

The main idea to prove hardness of $\GEM{2}{3}$ is to reduce from the problem \stretcha.
In \textsc{\stretcha}, we are given an arrangement of pseudolines (curves) in the plane and we are asked to decide whether  there exists a set of straight lines that has the same combinatorial pattern as the pseudoline arrangement, {see \Cref{fig:overview}(a) for an illustration and \cref{sec:definitions} for a formal definition.}
Given a pseudoline arrangement $L$, we construct a $2$-complex~$C$ {which} has a geometric embedding in $\R^3$ if and only if $L$ is stretchable.
\begin{figure}[p]
\centering
    \includegraphics[page = 5]{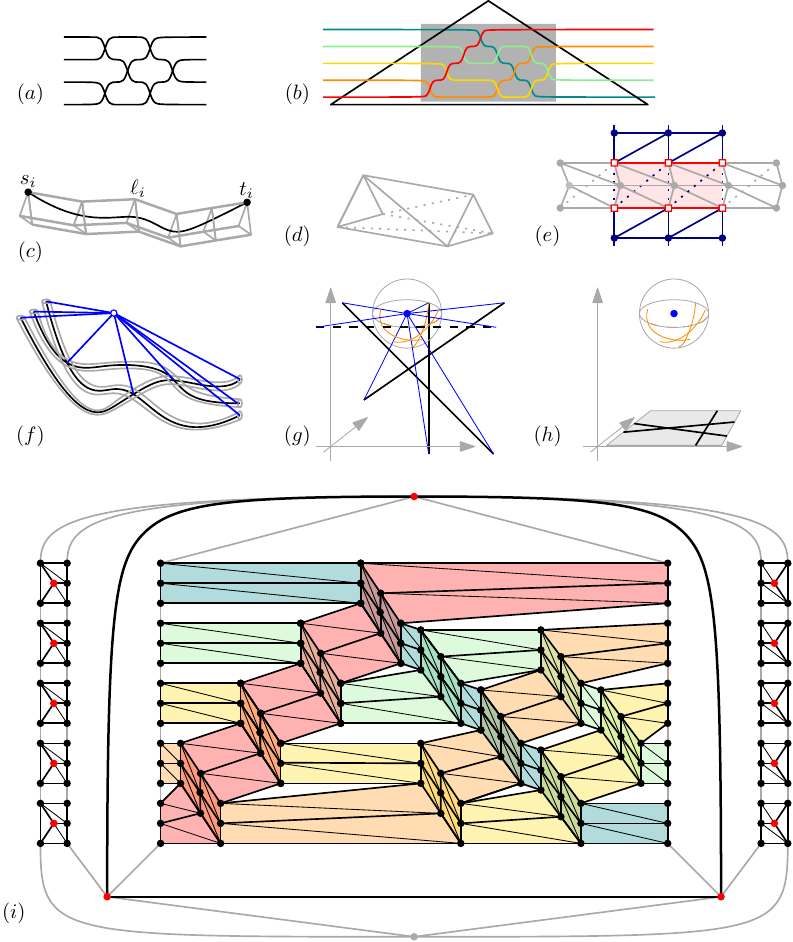}
\caption{(a) We start with a pseudoline arrangement $L$.
    (b) We add a triangle containing all intersections of $L$.
    (c) Each pseudoline is represented by a special edge that is surrounded by a tunnel.
    (d) Each tunnel consists of tunnel sections.
    (e) For the crossings of the special edges, we identify parts of the tunnels. 
    (f) We add an apex $u$ and insert triangles to the visible parts of the construction; we enhance the neighborhood of the apex to an essentially 3-connnected graph depicted in (i). Together, (e), (f) and (i) are crucial to enforce the correct combinatorics.
    (g) In the correctness proof, we use a small sphere around the apex and the projection  of each special edge onto the sphere.
    (h) We will argue that the combinatorics on the sphere are correct and then project the special edges onto a plane. 
    This will give a stretched arrangement.
    (i) The neighborhood graph of the apex $u$.}
    \label{fig:overview}
\end{figure} 
On a high level, our construction of $C$ goes along the following lines:
{We add a \htriangle that  contains all intersections of the pseudolines}, see \Cref{fig:overview}(b).
We place each pseudoline in $\R^3$ and replace it by a \emph{special edge} of the complex~$C$.
We surround the special edges by tunnels, see \Cref{fig:overview}(c) and (d). {For each crossing in $L$, } we glue the corresponding tunnel sections together, see \Cref{fig:overview}(e).
At last, we insert an {apex~$u$}  high above that is connected to all visible tunnel parts, see \Cref{fig:overview}(f)  
{and we insert additional objects in order to ensure that the neighborhood of $u$ is an essentially 3-connected graph, \Cref{fig:overview}(i).}

It is relatively straightforward to verify that if $L$ is stretchable, then the complex~$C$ embeds geometrically into $\R^3$.
{The other direction
requires more care and work:}
We show that a geometric embedding of $C$ induces a line arrangement with the same combinatorics as $L$.
{The idea of the proof is to consider a small sphere around the apex $u$ and to project its neighborhood and the special edges onto the sphere, see \Cref{fig:overview}(g). 
Because the neighborhood graph of $u$ is  essentially 3-connected by construction,  all its crossing-free drawings on the sphere are equivalent. This is an crucial property to show that each special edge lies in the projection of its tunnel roof (when restriction the attention to an interesting part within the \htriangle). 
We remark, that our proof does not show this explicitly. Instead, we establish some even stronger properties.
As a consequence, the projection of the tunnels have the intended combinatorics and thus also the special edges.
At last, we project the  arcs from the
sphere onto a plane, see \Cref{fig:overview}(h). 
In this way, we obtain a line arrangement with the same combinatorics as~$L$.
}

In order to show hardness of $\GEM{3}{3}$, we use a similar construction, in which  we ``fatten'' each triangle to a tetrahedron, by adding extra vertices.

We finally present a dimension reduction, i.e., we reduce $\GEM{k}{d}$ to $\GEM{k+1}{d+1}$.
Given a $k$-complex $C$, we create a $(k+1)$-complex $C^+$ that  contains $C$ and has two additional vertices $a$ and $b$.
Moreover, for each {subset} $e$ of $C$, $C^+$ has the additional {subsets} $e\cup\{a\}$ and $e\cup\{b\}$.
We prove that~$C$ geometrically embeds in~$\R^d$ 
if and only if~$C^+$  geometrically embeds in~$\R^{d+1}$.

In this way, we show that
distinguishing PL embeddable and geometrically embeddable complexes is \ER-complete.

%%%%%%%%%%%%%%%%%%%%%%%%%%%%%%%%%%%%%%%
\subsection{Existential Theory of the Reals}
\label{sec:ExReals}
%%%%%%%%%%%%%%%%%%%%%%%%%%%%%%%%%%%%%%%
The class of the existential theory of the reals \ER (pronounced as 
`exists R', `ER', or `ETR')
is a complexity class which has gained a lot of interest in recent years, specifically in the computational geometry community.
To define this class, we first consider the algorithmic problem \emph{Existential Theory of the Reals~(ETR)}. 
An instance of this problem consists of a sentence of the form
\[
\exists x_1, \ldots, x_n \in \R : \Phi(x_1, \ldots, x_n),
\]
where~$\Phi$ is a 
well-formed 
quantifier-free formula over the alphabet
$\{0, 1,  +, \cdot, \geq, >, \wedge, \vee, \neg\}$, and
the goal is to check whether this sentence is true. 
As an example of an ETR-instance, consider 
$\exists x,y \in \R\colon\Phi(x,y) = (x \cdot y^2 + x \geq 0) \wedge \neg(y < 2x)$, for which the goal is to determine whether there exist real numbers~$x$ and~$y$ satisfying the  formula~$\Phi(x,y)$.

The \emph{complexity  class \ER} is the family of all problems that admit a polynomial-time many-one reduction to ETR.
It is known that
\[
\textrm{NP} \subseteq \ER \subseteq \textrm{PSPACE}.
\]
The first inclusion follows from the definition of \ER. 
Showing the second inclusion was first established by Canny in his seminal paper~\cite{canny1988}. 
The complexity class \ER gains its  significance because a number of well-studied problems from different areas of theoretical computer science have been shown to be complete for this class.

Famous examples from discrete geometry are the recognition of geometric structures, such as unit disk graphs~\cite{mcdiarmid2013integer}, segment intersection graphs~\cite{matousek2014intersection}, \stretcha~\cite{mnev1988universality,shor1991stretchability}, and order type realizability~\cite{matousek2014intersection}.
Other \ER-complete problems are related to graph drawing~\cite{AnnaPreparation}, Nash-Equilibria~\cite{bilo2016catalog,garg2018etr}, geometric packing~\cite{etrPacking}, the art gallery problem~\cite{ARTETR}, non-negative matrix factorization~\cite{shitov2016universality}, polytopes~\cite{NestedPolytopesER,richter1995realization}, geometric linkage constructions~\cite{abel}, training neural networks~\cite{Training-Neural-Networks}, visibility graphs~\cite{cardinal2017recognition}, continuous constraint satisfaction problems~\cite{Reinier-CSP}, and convex covers~\cite{abrahamsen2021covering}.
The fascination for the complexity class stems not merely from the number of \ER-complete 
problems but from the large scope of seemingly unrelated \ER-complete problems.
We refer the reader to the lecture notes by Matou\v{s}ek~\cite{matousek2014intersection} and 
surveys by Schaefer~\cite{Schaefer2010} and Cardinal~\cite{CardinalSurvey} for more 
information on the complexity class~\ER.

%%%%%%%%%%%%%%%%%%%%%%%%%%%%%%%%%%%%%%%
\subsection{Definitions}
\label{sec:definitions}
%%%%%%%%%%%%%%%%%%%%%%%%%%%%%%%%%%%%%%%

\paragraph{Simplex.}
A  $k$-simplex $\sigma$ is a $k$-dimensional polytope which is the convex hull of its $k+1$ vertices~$V$, which are not contained in the same $(k-1)$-dimensional hyperplane.
Hence, a 0-simplex corresponds to a point, a 1-simplex to a segment, and a 2-simplex to a triangle etc.
The convex hull of any nonempty proper subset of $V$
is called a \emph{face} of~$\sigma$.
 A \emph{simplicial complex}  $K$ is a set of simplices satifying the following two conditions: (i) Every face of a simplex from $K$ is also in $K$.
(ii)  For any two simplices~$\sigma_1,\sigma_2\in K$ with a non-empty intersection, the intersection $\sigma_1\cap\sigma_2$ is a face of both simplices $\sigma _{1}$ and $\sigma _{2}$.
The purely combinatorial counterpart to a simplicial complex is an abstract simplicial complex, which we refer to simply as a \emph{complex}.

\paragraph{Complex.}
A \emph{complex} $C = (V,E)$ is a finite set $V$ together with a collection of subsets $E\subseteq 2^V$ which is closed under taking subsets, i.e., $e\in E$ and $e' \subseteq e$ imply that $e' \in E$.
A \emph{$k$-complex} is a complex where the largest subset contains exactly $k+1$ elements. We call a complex \emph{pure} if all (inclusion-wise) maximal elements in $E$ have the same cardinality.

For any vertex $v\in V$ in a $k$-complex $C=(V,E)$, the neighbourhood of $v$ gives rise to a lower dimensional complex $C_v:=(V',E')$, where $E':=\{e\setminus\{v\}\mid v\in e\in E\}$ and $V':=N(v)=\bigcup_{e\in E'} e$ are the \emph{neighbors} of~$v$.
Complexes are in close relation to Hypergraphs.

\paragraph{Hypergraphs.}
Hypergraphs generalize graphs by allowing edges to contain any number of vertices. 
Formally, a \emph{hypergraph}~$H$  is a pair $H = (V , E )$ where $V$ is a set of  vertices, and $E$ is a set of non-empty subsets of $V$ called \emph{hyperedges} (or edges). 
A $k$-uniform hypergraph is a hypergraph such that all its hyperedges contain exactly $k$ elements. 
Note that the maximal sets of a pure $k$-complex yield a $(k+1)$-uniform hypergraph and vice versa. Hence, $(k+1)$-uniform hypergraphs and pure $k$-complexes are in a straight-forward one-to-one correspondence.
A \emph{simplicial representation} of a $(k+1)$-uniform hypergraph is a geometric embedding of the corresponding complex.

\paragraph{Geometric embeddings.}
A \emph{geometric embedding} of a complex $C = (V,E)$ in $\R^{d}$ is a function $\varphi \colon V \rightarrow \R^d$ fulfilling the following two properties: (i) for every $e\in E$, $\overline \varphi(e):=
\conv(\{\varphi(v)\colon v\in e\})$ is a simplex of dimension $|e|-1$ and (ii) for every pair $e,e'\in E$, it holds that
\[ \overline\varphi(e) \cap \overline\varphi(e') = \overline\varphi(e\cap e').\]
 Note that if $\varphi$ is a geometric embedding, then $\{\overline\varphi(e): e\in E\}$ is a simplicial complex.
The problem \GEM{k}{d} asks whether a given  $k$-complex has a geometric embedding in $\R^{d}$.

\paragraph{Topological and PL embeddings.}
Consider a complex $C=(V,E)$.
In contrast to geometric embeddings, for PL or topological embeddings it is not sufficient to describe the mapping of the vertices $V$.
Choose $d'$ so large that $C$ admits a geometric embedding $\varphi': V\to \R^{d'}$, and define $S= \bigcup_{e\in E} \overline{\varphi'}(e)$.
We then say that an injective and continuous function $\varphi: S\to \R^d$ is a \emph{topological embedding} of $C$ in $\R^d$.
If furthermore for each $e\in E$, the image $\varphi(\overline{\varphi'}(e))$ is a finite union of connected subsets of $(|e|-1)$-dimensional hyperplanes, then $\varphi$ is a \emph{piecewise linear (PL) embedding}.
The problem \PLEM{k}{d} asks whether a given  $k$-complex has a PL embedding in $\R^{d}$.

\paragraph{Graph Drawings.}
A graph is a 1-complex. 
A graph is \emph{planar} if there exists a crossing-free drawing in the plane, i.e., a (topological) embedding in $\R^2$. As mentioned above, a graph has a topological embedding in $\R^2$ if and only if it has a geometric embedding in $\R^2$. A \emph{plane} graph is a planar graph together with a specified crossing-free drawing. 
By means of stereographic projection,  any graph that has a crossing-free drawing in the plane also has a crossing-free drawing on the sphere  and vice versa.
Two drawings of a graph (in the plane or on the sphere) are \emph{equivalent} if they can be transformed into one another by a homeomorphism (of the plane or the sphere). 
In particular, two equivalent drawings have the same set of faces. 
Consider a plane graph $G$ and let $D'$ be the specified drawing of $G$.
When talking about an (arbitrary) drawing $D$ of $G$, we always mean that~$D$ is equivalent to $D'$.

\paragraph{Stretchability.}
A \emph{pseudoline arrangement} is a family of curves that apart from `straightness' share similar  properties with a line arrangement. 
More formally,  a \emph{(Euclidean) pseudoline arrangement} is an arrangement of $x$-monotone curves in the Euclidean plane such that any two meet in exactly one point.
% every pair of pseudolines meets exactly once in a crossing point.
In fact, each pseudoline arrangement can be  encoded by a \emph{wiring diagram}; see also \cref{fig:pseudolines}.
A pseudoline arrangement is \emph{stretchable} if it is combinatorially equivalent to an arrangement of straight-lines, i.e., if the arrangements can be transformed into one another by a homeomorphism of the plane. \stretcha denotes the algorithmic problem of deciding whether a given pseudoline arrangement is stretchable. In a seminal paper, Shor~\cite{shor1991stretchability} proved that \stretcha  is complete for the existential theory of the reals; for a stream-line exposition of this result see the expository paper by Matou\v{s}ek \cite{matousek2014intersection}.

%%%%%%%%%%%%%%%%%%%%%%%%%%%%%%%%%%%%%%%
\subsection{Pitfalls}
\label{sub:pitfalls}
%%%%%%%%%%%%%%%%%%%%%%%%%%%%%%%%%%%%%%%
While the general proof ideas are fairly straightforward, our arguments in \Cref{sec:thm} may at first glance appear a bit tedious. 
In the following, we highlight one of the appearing challenges.
It is easy to see that each special edge lies inside its tunnel in any geometric embedding.
It follows that the projection of the special edge lies also inside
the projection of 
the tunnel on the sphere centered at the apex.
Furthermore, we know that the roof of the tunnels are seen by the apex.
One may be tempted to (directly) conclude that 
the projection of 
the special edge is thus also contained in
the projection of
the roof;
the underlying thought being that 
the  projection of 
the tunnel bottom lies below the tunnel roof in the geometric representation and thus the projection of the tunnel bottom is contained in 
the projection of
the tunnel roof.
Yet, 
the latter is not true in general,
 as can be seen in \Cref{fig:BottomSeen}.
 In the figure, the tunnel bottom  is not covered by the roof.
We (implicitly) show  that the projection of the
special edge lies inside the projection of the roof by establishing some even stronger topological and geometric properties.

\begin{figure}[htb]
\centering
\includegraphics[page = 1]{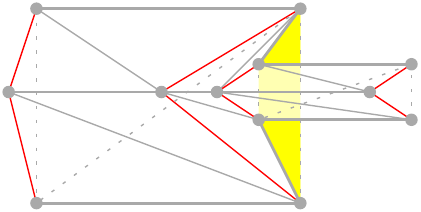}
\caption{
From the perspective of $u$, the tunnel bottom is not always hidden below the tunnel roof. From the three tunnel sections displayed, the bottom (yellow) of the middle section is partially visible from the apex.
}
\label{fig:BottomSeen}
\end{figure}
%

% \clearpage
% \newpage
%%%%%%%%%%%%%%%%%%%%%%%%%%%%%%%%%%%%%%%
\section{The Proof}
\label{sec:thm}
%%%%%%%%%%%%%%%%%%%%%%%%%%%%%%%%%%%%%%%

In this section, we prove \cref{thm:main}. Our proof consists of the following three parts.

\begin{enumerate}[a)]
\setlength\itemsep{0pt}
\item Establishing \ER-membership (\cref{sec:membership}: \cref{lem:membership}).
\item Showing \ER-hardness in~$\R^3$, i.e., of \GEM{2}{3} and \GEM{3}{3}
(\cref{sec:2dim}: \cref{thm:TwoThree,lem:threethree}).
\item Reducing \GEM{k}{d} to \GEM{k+1}{d+1} (\cref{sec:dim-reduction}: \cref{lem:OneUp}).
\end{enumerate}
Together \cref{lem:membership,thm:TwoThree,lem:threethree,lem:OneUp} prove \Cref{thm:main}.

%%%%%%%%%%%%%%%%%%%%%%%%%%%%%%%%%%%%%%%
\subsection{Membership}
\label{sec:membership}
%%%%%%%%%%%%%%%%%%%%%%%%%%%%%%%%%%%%%%%
In this subsection, we show \ER-membership of \GEM{k}{d}.
Note that this is essentially folklore~\cite{vcadek2014extendability}.
We present a proof for the sake of completeness.
\begin{lemma}\label{lem:membership}
For all $k,d\in\N$, the decision problem \GEM{k}{d} is contained in \ER.
\end{lemma}

\begin{proof}
In order to show membership in \ER, we use the following characterization by Erickson, Hoog and Miltzow~\cite{SmoothingGap}:
A problem~$P$ lies in \ER if and only if there exists a real verification algorithm $A$ for $P$
that runs in polynomial time on the real RAM. In particular,  for every yes-instance~$I$ of $P$ there exists a polynomial sized witness~$w$ such that
$A(I,w)$ returns yes, and 
for every no-instance~$I$ of $P$ and any witness~$w$,
$A(I,w)$ returns no.
In contrast to the definition of the complexity class NP,  we also allow witnesses that consist of real numbers. 
Consequently, we execute $A$ on the real RAM as well.

It remains to present a real verification algorithm for \GEM{k}{d}.
While the witness describes the coordinates of the vertices,  the algorithm checks for intersections between any two simplices.
Note that each simplex is a convex set and the intersection of convex sets is a convex set as well. 
For any simplex $S$ with $n$ vertices, we can efficiently determine $n$ linear inequalities and at most one linear equality that together describe $S$. Then checking for intersections can be reduced to a linear program, which is polynomial time solvable. This finishes the description of the real verification algorithm.
\end{proof}

%%%%%%%%%%%%%%%%%%%%%%%%%%%%%%%%%%%%%%%
\subsection{Hardness in three dimensions}
\label{sec:2dim}
%%%%%%%%%%%%%%%%%%%%%%%%%%%%%%%%%%%%%%%
This section is dedicated to proving \cref{thm:main} for $d=3$ and $k\in\{2,3\}$.
The crucial part lies in the case $k=2$.
For the benefit of the reader, we include a glossary in \Cref{tab:glossary}.
%%%%%%%%%%%%%%%%%%%
\begin{theorem}\label{thm:TwoThree}
 The decision problem \GEM{2}{3} is \ER-hard.
\end{theorem}
%%%%%%%%%%%%%%%%%%%
\begin{proof}

We reduce from the \ER-hard problem \stretcha, as described in \cref{sec:definitions}.
Let $L$ be an arrangement of $n$ pseudolines in the plane. Every pseudoline arrangement has a representation as a wiring diagram in which each pseudoline is given by a monotone curve consisting of $2n-1$ sections. 
For an illustration consider \cref{fig:pseudolines}; each section could be represented by a segment, however for a visual appealing display, the bend points are rounded.
We add a pseudoline $\ell_{0}$ that intersects all pseudolines in the beginning, see \cref{fig:pseudolines}, and call the resulting pseudoline arrangement $L^*$. Note that $L^*$ is stretchable if and only if $L$ is stretchable. For later reference, we endow a natural orientation upon each pseudoline from left to right.
\begin{figure}[ht]
\centering
    \begin{subfigure}[t]{0.45\textwidth}
\centering
\includegraphics[page = 1]{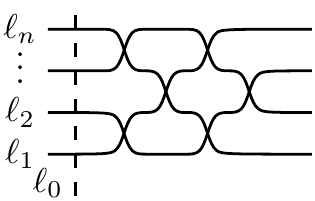}
\caption{A pseudoline arrangement $L^*$.}
\label{fig:pseudolines}
\end{subfigure}
\hfil
\begin{subfigure}[t]{0.45\textwidth}
\centering
\includegraphics[page = 2]{figures/Construction42.pdf}
\caption{The crossing diagram contains an additional \htriangle.}
\label{fig:helperTriangle}
\end{subfigure}
\caption{Adding an extra pseudoline $\ell_0$ and the \htriangle $\triangle$ to the construction.}
\end{figure} 
In the following, we construct a 2-complex $C=(V,E)$ that allows for a geometric realization if and only if $L^*$ (and thus $L$) is stretchable. In order to define $C$, we add a \htriangle $\triangle$ to our arrangement that intersects the pseudolines of $L^+$ as illustrated in \cref{fig:helperTriangle}.

\paragraph{Construction of the 2-complex.}
In order to define $C$, we use the intended geometric embedding. We will refer to the subsets in $C$ as vertices, edges, and triangles depending on whether they contain one, two or three elements. 
The construction has 
five
steps.
\medskip

%Tunnels
\textbf{In the first step}, we place the pseudolines and the \htriangle $\triangle$ in 3-space.  Each pseudoline $\ell_i$  lies in the plane $z=i$ such that an observer high above (at infinity) sees the wiring diagram. Similarly, we place the segments of the \htriangle $\triangle$ in 3-space such that it lies in the plane $z=n+1$.  Note that no two pseudolines intersect. Therefore,  we can surround each lifted pseudoline by a triangulated sphere which we call a \emph{tunnel}; see also \cref{fig:tunnel}.
The tunnel $\Tall_i$ of~$\ell_i$ is formed by $2n+3+i$ sections; 
{later, we will be particularly interested in a part of a tunnel, denoted by $\Tint_i$, in which the first two and last two sections are removed}. 
Each section consists of six triangles forming a triangulated triangular prism as illustrated in \cref{fig:section}. We close the tunnel with triangles at the ends and  think of the \emph{bottom} side of the prism to lie in the plane~$z=i-\nicefrac{1}{2}$ (for now).
The remaining part of the tunnel, i.e., the tunnel without its bottom, constitutes the \emph{roof}, see \cref{fig:tunnelAbove}. 
The roof contains three disjoint paths of length $2n+4+i$.
The edges and vertices on the boundary of the bottom and the roof form the \emph{left} and \emph{right roof path}, respectively, when deleting the edges of the closing triangles. 
The remaining vertices induce the \emph{central roof path}. The three roof paths are thickened in \cref{fig:tunnel,fig:tunnelAbove}.

\begin{figure}[htb]
\centering
\begin{subfigure}[b]{0.35\textwidth}
\centering
\includegraphics[page = 3]{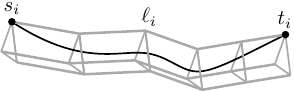}
\caption{A tunnel viewed from side.}
\label{fig:tunnel}
\end{subfigure}
\hfill
\begin{subfigure}[b]{0.25\textwidth}
\centering
\includegraphics[page = 5]{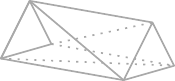}
\caption{A tunnel section.}
\label{fig:section}
\end{subfigure} 
\hfill
\begin{subfigure}[b]{0.35\textwidth}
\centering
\includegraphics[page = 6]{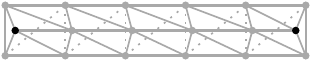}
\caption{A tunnel viewed from above.}
\label{fig:tunnelAbove}
\end{subfigure} 

\caption{First step in the construction of the complex $C$ -- tunnel construction.}
\label{fig:Construction1}
\end{figure}

Note that we do not add a tunnel for the \htriangle.
We distribute the sections along $\Tall_i$ to edges and crossings of the crossing diagram as follows: Generally, we associate one section per edge and one section per crossing of two pseudolines. Moreover, we associate one extra section of~$\Tall_j$ to a crossing of $\ell_i$ and $\ell_j$ whenever $i<j$.
In order to represent the pseudoline $\ell_i$, we insert a \emph{special edge} $e_i$  between the two top vertices on either end of the tunnel; for later reference, we denote the start vertex by  $s_i$ and the end vertex by $t_i$. The special edge~$e_i$ is intended to lie inside the tunnel. 
\medskip

%Gluing Tunnels
\textbf{In the second step}, we identify parts of the tunnels. To this end, consider the tunnel sections assigned to a crossing of a pseudoline $\ell_i$ with $\ell_j$, $i<j$. Recall that we assigned one section of $\Tall_i$ and two sections of $\Tall_j$ to the crossing. We identify the four triangles in the bottom of the two sections of $\Tall_j$ with the four triangles in the roof of one section of $\Tall_i$ as indicated in \cref{fig:gluing}. Note that we hereby identify six vertices, four of which belong to a left or right roof path of both, $\Tall_i$ and $\Tall_j$. 
\medskip

\begin{figure}[htb]
\centering
\includegraphics[page = 7]{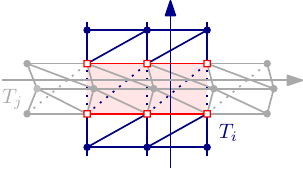}
\caption{Second step in the  construction of the complex $C$: gluing of tunnel parts.}
\label{fig:gluing}
\end{figure}
%

%Apex
\textbf{In the third step}, we add a new vertex to the construction that we call the \emph{apex} and denote it by~$u$. We think of $u$ as the observer high above (at infinity) and insert a triangle defined by~$u$ and the vertices of every tunnel edge that is \emph{visible} from $u$. Note that every roof section that is neither glued in a crossings nor hidden by the \htriangle is visible. Moreover, no bottom  of any tunnel is visible in the intended geometric embedding.
\medskip

%Enhance apex neighborhood
\textbf{In the fourth step}, we enhance the $1$-complex induced by the neighborhood $N(u)$ of the apex~$u$
such that it corresponds to an essentially $3$-connected planar graph~$G^+$. We call a graph \emph{essentially $3$-connected} if it is a subdivision of a $3$-connected graph. 
With the description so far, the $1$-complex corresponds to the graph $H$ depicted in black in \cref{fig:3-connected}.
\begin{figure}[htb]
	\centering
	\includegraphics[page = 8]{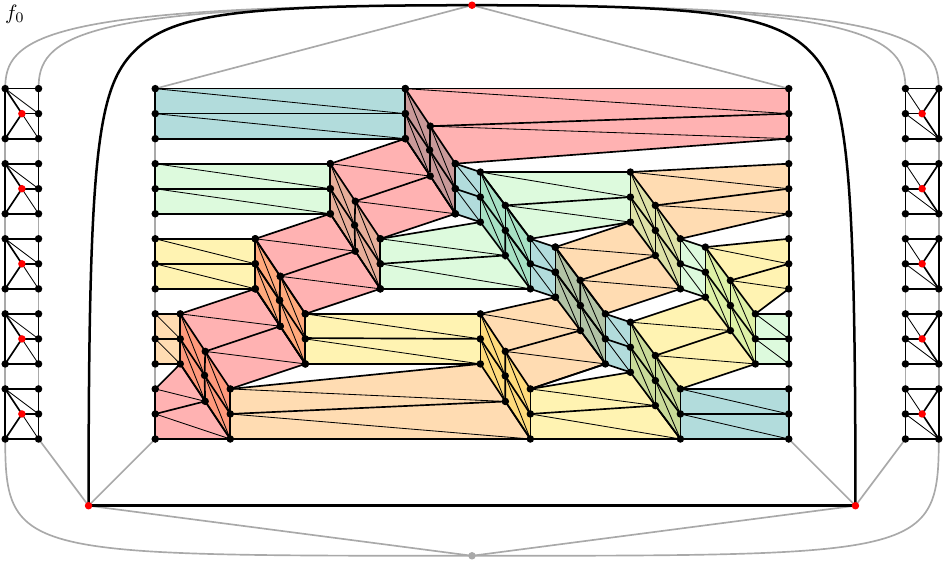}
	\caption{
		Third and fourth step in the construction of the complex $C$: neighborhood of the apex~$u$. The graph $H$ depicted in black after the third step. Together with the gray edges, the graph is a candidate for the essentially 3-connected plane graph $G^+$ and it subgraph $G$ inside $\triangle$.}
	\label{fig:3-connected}
\end{figure} 
To construct $G^+$, we make use of the following fact.

\begin{claim}\label{clm:enhancing}
    For every plane graph $G_1=(V_1,E_1)$, there exists an essentially 3-connected plane graph $G_2=(V_2,E_2)$ such that $G_1$ is a subgraph of $G_2$ and any
    straight-line drawing~$D_1$ of $G_1$ in the plane can be extended to a straight-line drawing of $G_2$. Moreover, if the maximum face degree of $G_1$ is
    $k$,
    then the size of $G_2$ 
    can be bounded by $|V_2|+|E_2|\leq  O(k|V_1|)$.
\end{claim}
\begin{proof}
    In order to construct $G_2$, we start with a drawing of $G_1$. 
    {First, we ensure that $G_1$ has at least 4 vertices. Second, we guarantee 2-connectedness by inserting edges; the edges are represented by potentially non-straight curves. To this end, we iteratively insert edges between different connected components; this ensures connectedness. Afterwards we insert one vertex in each face and insert one edge to each incident vertex of this face; this ensures 2-connectedness.
    Third, we triangulate each face by repeating the last step: we insert a vertex into each face and insert an edge to every vertex incident to this face. By 2-connectedness,} we obtain a triangulation $T$ on at least four vertices and hence, a 3-connected plane graph. Lastly, we subdivide each new edge so that the number of subdivision vertices equals to the degree of the face in $G_1$ in which the edge has been inserted. The resulting graph is $G_2$ and by construction essentially 3-connected.
    {Note that $T$ has  $O(|V_1|)$ vertices. Therefore, $G_2$ has $O(k|V_1|)$ edges.}
    
    It remains to show that any straight-line drawing $D_1$ of the plane graph $G_1$ can be extended to a straight-line drawing of the planar graph $G_2$. Because $D_1$ and $G_1$ have the same set of faces, we can insert the additional edges of $T$ by polylines. Following the face boundary, the number of bends on each edge is upper bounded by the degree of the face of $G_1$ in which it is inserted. Hence, we can easily extend $D_1$ to a straight-line drawing of  $G_2$.
\end{proof}

Let $G^+:=G_2$ be an essentially 3-connected plane graph guaranteed by \cref{clm:enhancing} for the case that $G_1=H$. {Note that $G_1$ has $O(n^2)$ vertices and edges, and every face has degree $O(n)$. Hence, the size of $G_2$ is in $O(n^3)$.} {We denote the outer face of $G_2$ by $f_0$.}
 The reader is invited to think about the far more sparse graph depicted in \cref{fig:3-connected}, which also serves as a candidate for~$G^+$. 
Indeed, the depicted graph also fulfills all properties necessary for our construction; {however, not all properties  of \Cref{clm:enhancing}.}
For example, the depicted graph is even $3$-connected. The proof of this is straightforward, but a bit tedious.
Thus, we leave it as an exercise to the interested reader to check that the graph remains connected even after the deletion of any two vertices or alternatively, that any pair of vertices is connected by three disjoint paths.

Later, the subgraph $G$ of $G^+$ that is induced by  all vertices of $\bigcup_i \Tint _i$  will be of particular interest; in \cref{fig:3-connected}, these vertices (and their convex hull) lie inside the \htriangle $\triangle$. Recall that $\Tint_i$ denote the part of the tunnel $\Tall_i$ obtained by deleting the first two and last two sections.

It is a well-known fact that all (straight-line or topological) planar drawings of a 3-connected planar graph on the sphere  are equivalent~\cite{imrich1975whitney}; for a definition of equivalent drawings consult  \cref{sec:definitions}.
Consequently, the result extends to \textit{essentially} $3$-connected graphs as it also holds for topological drawings. For later reference, we note the following.
\begin{claim}
    \label{clm:ExtendG}
The planar graph $G^+$ is essentially  $3$-connected. Therefore, all crossing-free drawings of~$G^+$ on a sphere are equivalent.
Furthermore, any straight-line drawing of $H$ in the plane can be extended to a straight-line drawing of $G^+$.
\end{claim}
We ensure that the neighborhood complex of $u$ is the underlying planar graph of $G^+$, i.e, for each edge of $G^+$ not present in $H$, we insert a triangle {formed by the vertices of this edge together with $u$} and call the resulting complex $\overline C$.
\medskip

%Duplicate
\textbf{In the fifth and last step}, our final complex $C$ consist of two copies of $\overline C$ in which the apex vertices are identified. We will later use these two copies in order to guarantee that in any geometric embedding the apex lies {outside of all tunnels} for one copy of $\overline C$.
This finishes the construction of the complex~$C$.

 \paragraph{Time Complexity.}
In order to verify that the construction shows \ER-hardness,  we argue that it has a running  time that is polynomial in the size of the input.
To this end, note that a pseudoline arrangement with $n$ pseudolines can be  described by the order of the $O(n^2)$ crossings. Thus, the input size is $N = O(n^2)$.  After adding the \htriangle and $\ell_0$, the crossing diagram still has a size in $O(n^2)$. 
It is easy to see that our construction has a size proportional to $N^{3/2}$: 
For each segment and crossing of the diagram, we build a constant size construction involving the apex. {Moreover, we add a triangle for every (additional) edge in $G^+$; recall that $G^+$ has size $O(n^3)$.}
Consequently, the total construction has size $O(n^3) = O(N^{3/2})$.
{We remark, that a more careful choice of $G^+$, as in \cref{fig:3-connected}, yields a construction that is linear in $N$.}

\paragraph{Correctness.}

It remains to show that the pseudoline arrangement $L$ is stretchable if and only if~$C$ has a geometric embedding in $\mathbb R^3$.

\paragraph{Correctness I: Stretchability implies
Embedabbility.}
If $L$ is stretchable, it is relatively straight-foward to construct a geometric embedding of $C$.
\begin{claim}
If $L$ is stretchable, then $C$ has a geometric embedding.
\end{claim}
\begin{proof}
The construction of the geometric embedding goes along the same lines as the construction of~$C$.
Consider the stretched line arrangement~$L'$ equivalent to the pseudoline arrangement~$L$, we construct $L''$, by adding $\ell_0$ to the left of all crossings of $L'$.

    Next we show how to add the \htriangle~$\triangle$ such that the intersection pattern is as depicted in the crossing diagram.
    Note that we merely have to construct the triangle in a way that all intersections of $L''$ are contained in it and the vertices of the \htriangle lie in the correct faces.
    For an illustration of the following argument see \Cref{fig:insert-triangle}.
       \begin{figure}[b]
    \centering
    \includegraphics[page = 2]{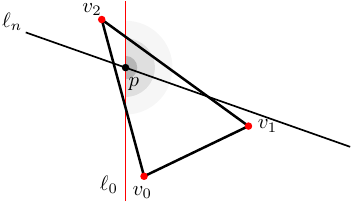}
    \caption{Illustration for the proof of \cref{clm:ExtendG} concerning the insertion of $\triangle$. By scaling $L''$ while keeping point $p$ and the \htriangle fixed, we can ensure that all intersections lie inside the \htriangle eventually.}
    \label{fig:insert-triangle}
    \end{figure}
  The \htriangle has three vertices~$v_0,v_1,v_2$.
By construction, the two vertices~$v_0,v_1$ are supposed to lie
in the bottom cell of~$L^*$ which is bounded by $\ell_0$ and $\ell_n$ and potentially other pseudolines.
    Similarly, $v_2$ is in the upper cell bounded exactly by $\ell_0$ and $\ell_n$. 
Let us fix some triangle $\triangle$ that contains the intersection point~$p$ of $\ell_0$ and $\ell_n$ with the properties above.
Consider the (gray) region that contains all intersections of~$L''$. 
By scaling $L''$ while fixing $p$ and the triangle $\triangle$, the gray region becomes arbitrarily small and is eventually contained in $\triangle$.
This shows that the \htriangle can be added as wished.

    Afterwards, we lift the lines and the \htriangle in $3$-space, construct the tunnels, and enhance $H$ to $G^+$. For the latter step, we use the fact that  we can extend any straight-line drawing of $H$ to a straight-line drawing of $G^+$ by \Cref{clm:ExtendG}.
    (In case that we want to use $G^+$ as in \Cref{fig:3-connected}, we need to use the fact that the drawing of $H$ comes from the stretched line arrangement and thus the faces are convex.)
    Afterwards, we insert  the apex $u$ (high enough) above, make a second copy by taking the mirror image, and identify the apices.
    This yields a geometric embedding of~$C$.
\end{proof}

\paragraph{Correctness II: Embedabbility implies Stretchability.}
The reverse direction is more involved. Let $\varphi$ denote a geometric embedding of~$C$. 
To show that $L$ is stretchable, we start with a collection of crucial properties.

By definition,  each tunnel forms a  closed topological sphere; all of which are pairwise disjoint.
By a generalization of the Jordan curve theorem, also known as the Jordan–Brouwer separation theorem, 
any topological embedding of a $(d-1)$-sphere in $\R^d$ splits the space into two components~\cite{JordanBrouwer}; we refer to the bounded component of a tunnel or any other topological sphere as its \emph{inside} and to the unbounded component as its \emph{outside}.
\begin{claim}\label{clm:apexAtInfinity}
   There exists a copy $\overline C$ in $C$ such that
    the apex $u$ lies outside all tunnels of $\varphi(\overline C)$.
\end{claim}
\begin{proof}
 Note that the apex $u$ lies inside at most one tunnel of $C$; otherwise, among all tunnels containing $u$, the innermost separates $u$ from the outermost. A contradiction to the fact that shares and edge with the vertices of all tunnels. Consequently, for at least one copy of $\overline C$ contained in $C$, the apex $u$ lies outside of each tunnel.
\end{proof}

From now on, we focus on the geometric embedding $\varphi(\overline C)$ of this $\overline C$ and do not make further use of the other copy.

\begin{claim}\label{clm:inTunnel}
    In $\varphi(\overline C)$,  the special edge $e_i$ lies inside its tunnel $\Tall_i$ for all $i$.
\end{claim}
\begin{proof}
 Consider the vertex $s_i$ of $e_i$ and the tunnel $\Tall_i$ forming a closed topological sphere incident to~$s_i$. 
 By construction, all roof edges of the first section of $\Tall_i$ are visible by  the apex $u$ and thus form triangles with $u$. Note that the edges non-incident to $s_i$ contain a cycle, see \cref{fig:ends2}.

\begin{figure}[htb]
\centering
\begin{subfigure}[b]{0.45\textwidth}
\centering
\includegraphics[page = 19]{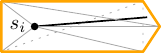}
\caption{A top view of first section.}
\label{fig:ends2}
\end{subfigure}
\hfil
\begin{subfigure}[b]{0.45\textwidth}
\centering
\includegraphics[page = 26]{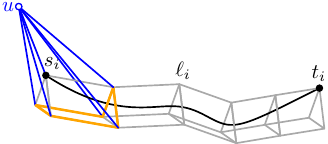}
\caption{The tent over $s_i$.}
\label{fig:RedinTunnel}
\end{subfigure}
\caption{Illustration for the proof of \cref{clm:inTunnel}: the spheres surrounding $s_i$.
}
\label{fig:Construction2}
\end{figure}

 The triangles between the apex $u$ and this cycle  form a \emph{tent} on top of  $\Tall_i$, see \cref{fig:RedinTunnel}. In particular, the tent together with the tunnel roof of the first section form another sphere incident to $s_i$. 
 By \cref{clm:apexAtInfinity}, $u$ is outside the tunnel.
 If $e_i$ started towards the outside of the tunnel, then it would be \emph{trapped} inside the tent-sphere which is not incident to $t_i$. 
 Consequently, $e_i$ lies inside the tunnel. 
\end{proof}

Consider a special edge $e_i$ and the wedge $W_i$ defined by rays originating at $u$ through points in~$e_i$, see  \cref{fig:aboveBelow}. We say a point $p$ of $(W_i\setminus e_i)$ 
lies \emph{above} $e_i$ if the segment $up$ does not intersect~$e_i$; we write $p>e_i$. Similarly, $p$ lies \emph{below} $e$ if the segment $up$ does intersect  $e_i$; we write $p<e_i$.

\begin{figure}[htb]
\centering
\includegraphics[page = 20]{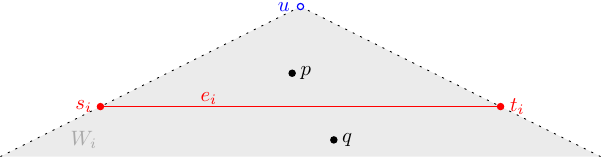}
\caption{The wedge $W_i$ of $e_i$ containing a point $p$ above $e_i$ and a point $q$ below $e_i$.}
\label{fig:aboveBelow}
\end{figure}

\begin{claim}\label{clm:aboveBelow}
Consider a special edge $e_i$.  For every two points $p,q$ on $W_i\setminus e_i$ that belong to a common triangle in $\Tall_i$, the above/below-relation is consistent, i.e.,  $p<e_i\iff q<e_i$.
\end{claim} 
\begin{proof}
Suppose for a contradiction that $p>e_i$ and $q<e_i$. By visibility of $s_i$ and $t_i$, $p$ lies neither on the segment $us_i$ nor on the segment $ut_i$. Hence,  the segment $pq$ intersects $e_i$ in an inner point, i.e., it intersects $e-\{s_i,t_i\}$, see \cref{fig:aboveBelow}. However, no triangle of $C$ contains $e_i$ and hence $e_i-\{s_i,t_i\}$ and $pq$ must have an empty intersection. A contradiction.
\end{proof}

Now, we consider a small sphere $S$ around the apex $u$ which has no other vertex inside, see \cref{fig:apexA}.
For each triangle in $\varphi(\overline C)$ containing $u$, the intersection with the sphere $S$ yields an arc of a great circle. Consequently, we obtain a crossing-free drawing $D^+$ of the essentially 3-connected planar graph $G^+$ (the neighborhood complex of $u$)
with arcs of great circles on $S$. 
Moreover, for each special edge $e_i$, we consider the projection $a_i$ of the (artificial) triangle $\{s_i,t_i,u\}$ onto the sphere $S$; note that $\{s_i,t_i,u\}$ is not a triangle of the complex, see \cref{fig:apexB}. 
\begin{figure}[b]
\centering
\begin{subfigure}[t]{0.3\textwidth}
\centering
\includegraphics[page = 24]{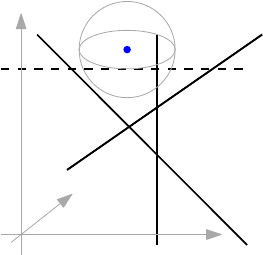}
\caption{}
\label{fig:apexA}
\end{subfigure}
\hfil
\begin{subfigure}[t]{0.3\textwidth}
\centering
\includegraphics[page = 22]{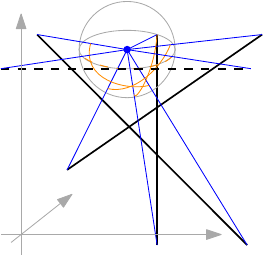}
\caption{}
\label{fig:apexB}
\end{subfigure}
\hfil
\begin{subfigure}[t]{0.3\textwidth}
\centering
\includegraphics[page = 23]{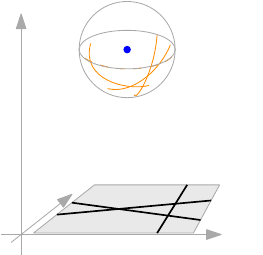}
\caption{}
\label{fig:apexC}
\end{subfigure}
\hfil
\caption{Obtaining a line arrangement from a geometric embedding.}
\label{fig:apex}
\end{figure}
This yields a set of arcs~$\A$.
In the remainder, we show that $\A$ has the same combinatorics/intersection pattern as the pseudoline arrangement~$L^*$, i.e., the order of intersections along each arc/pseudoline are the same. Afterwards, we project $\A$ to a plane as illustrated in \cref{fig:apexC} and obtain a line arrangement proving that $L^*$ (and thus $L$) is stretchable.

Let $D$ denote the restriction of the drawing $D^+$ to the graph $G$. By \cref{clm:ExtendG}, all crossing-free drawings of $G^+$ on the sphere are equivalent. This implies the following fact:
\begin{claim}\label{clm:interestingPart}
In $D^+$, the \htriangle separates the subdrawing $D$ from the vertices $\{s_j,t_j\mid j=0,\dots, n\}$. In particular, by convexity,  the convex hull of any vertex subset $U$ in $D$ is  contained in the \htriangle.
\end{claim}

For a tunnel, we call the intersection of two consecutive sections a \emph{\sloop}. Note that each \sloop consists of three edges, two of which belong to the roof of the tunnel and form a \emph{\rafter};  the remaining edge belongs to the bottom of the tunnel.

\begin{claim}\label{clm:tunnelRoof}\label{clm:order}
 $D\cup \A$ fulfills the following properties for all $i$:
\begin{enumerate}%[i)]
\setlength\itemsep{0pt}
\item \label{item:i} The arc $a_i$ intersects each \rafter of $\Tint_i$ exactly one time. 
\item \label{item:ii} The intersection points appear in the correct order along $a_i$, i.e., they follow the natural order along the tunnel. 
\item \label{item:iii} The faces of a roof section in $\Tint _i$ cover an interval of $a_i$.
\end{enumerate}
\end{claim}

\begin{proof}
Consider a \sloop  $\loop$ of tunnel $\Tall_i$ and its (projected) vertices in $D$. 
We first show that the arc~$a_i$ intersects at least one (potentially artificial)
segment connecting the vertices of
$\loop$ in $D$.
To this end, {we may assume that $o$ is disjoint from $s_i$ and $t_i$ and} consider the restriction of $\varphi(\overline C)$ to the plane~$P_i$ containing $u$ and $e_i$.
{By \cref{clm:apexAtInfinity,clm:inTunnel}, the apex $u$ is outside of the tunnel~$\Tall_i$ while~$e_i$ is inside $\Tall_i$. Hence, there exist curves~$\gamma_a,\gamma_b\subset \Tall_i$  from $s_i$ to $t_i$ on $P_i$ that together enclose~$e_i$ and separate $e_i$ from $u$, see \cref{fig:LoopIntersection}.}

\begin{figure}[htb]
	\centering
		\includegraphics[page = 13]{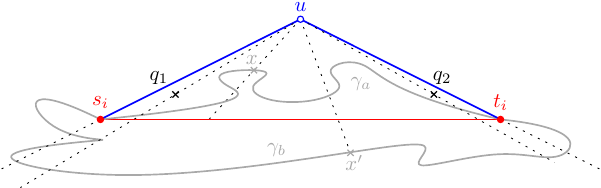}
		\caption{An illustration for the proof of \cref{clm:order}: the curve $\gamma_a$ in $P_i$ separates $e_i$ from $u$.}
		\label{fig:LoopIntersection}
\end{figure}

Because $s_i$ and $t_i$ are visible from $u$, the curves  intersect the segments $s_i u$ and $t_i u$ exactly in~$s_i$ and $t_i$, respectively. We denote the curve  contained in the triangle of~$u,s_i,t_i$ by~$\gamma_a$. The curve~$\gamma_a$  contains a point $x\in\loop$ because  $s_i$ and $t_i$ lie in different components of $\Tall_i-\loop$. Hence, the ray supporting $ux$ intersects $e_i$, i.e., $x>e_i$.

Next we show that  if  $\loop$ belong to $\Tint _i$, then the curve $\gamma_b$ intersects $\loop$ in a point below $e_i$.
Analogous to $\gamma_a$, the curve~$\gamma_b$ also contains a point $x'\in\loop$. 
It remains to show that $x'$ is below $e_i$.
Note that this does not hold, only in the case that $x'$ is to the left of the ray~$us_i$  or to the right of the ray~$ut_i$;  otherwise, $\gamma_b$ is below $e_i$.
By \Cref{clm:interestingPart}, the projection of the \htriangle separates $D$ from the 
vertices $s_i,t_i$ in $D^+$.
Consequently, the \htriangle $\triangle$ and $P_i$
intersect in two points $q_1$ and~$q_2$.
The points $q_1$ and $q_2$ must be above $e_i$,
because $u$ forms a triangle with every edge of $\triangle$.
By \cref{clm:interestingPart} and the fact that the vertices of \loop are contained in $D$, the vertices of \loop and their convex hull are inside $\triangle$ in $D$, {see \cref{fig:helper-triangle}.} Note that \loop is not completely contained in~$D$ because it contains an edge of the tunnel bottom.

\begin{figure}[htb]
	\centering
	\includegraphics[page=5]{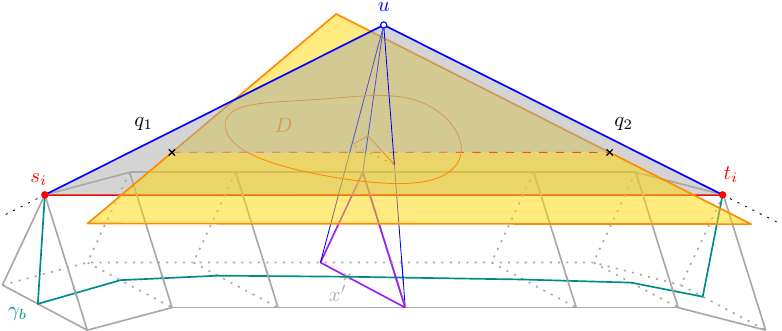}
	\caption{An illustration for the proof of \cref{clm:order}. The wedge $W_i$ and the \htriangle $\triangle$ (yellow) intersect in the segment $q_1q_2$.
	We drew the projection of $D$ onto the the \htriangle.
	In $D$, the vertices of the sloop $o$ (purple) are contained in $\triangle$. Therefore, 
	the point $x'$ lies between the rays $uq_1$ and $uq_2$ on $W_i$ and thus below $e_i$. 
	}
	\label{fig:helper-triangle}
\end{figure}
This implies that $x'$ is in a part of $\gamma_b$ 
that is bounded by the rays $uq_1$ and $uq_2$, see \cref{fig:LoopIntersection}.
Therefore, $x'$ lies below $e_i$ and is thus invisible from~$u$.
It follows that $x'$ belongs to an invisible edge of $\loop$, i.e., $x'$ is contained in the edge of $o$ belonging to the tunnel bottom; we denote it by $b_\loop$.
By \cref{clm:aboveBelow}, no point of  $b_\loop$ lies above~$e_i$. Consequently, $x>e_i$ belongs to the tunnel roof, i.e., the \rafter of $\loop$.

Moreover, no triangle of $\Tint_i$
incident to $b_\loop$ lies above $e_i$. 
Note that for each triangle in the bottom of $\Tint _i$, there exists a \sloop $\loop$ such that $b_\loop$ is incident to it. Hence, no point of the bottom of~$\Tint _i$  lies above $e_i$, i.e.,  we obtain the following \emph{Property 1}: the curve~$\gamma_a$ does not contain points of the bottom of the tunnel~$\Tint _i$.

It remains to argue that there exists no further intersection point {if $\loop\subset \Tint_i$}. 
We show that $(\gamma_a\cup\gamma_b)\cap \loop=\{x,x'\}$. Suppose that there exists a further point $x''\in (\gamma_a\cup\gamma_b)\cap \loop$.
Then {$\loop$ must be contained in $P_i$ and
we may assume that $x,x''$ are vertices of $\loop$. 
However, then the edge $xx''$ intersects the closed curve formed by concatenating the segments $us_i$, $s_it_i=e_i$, $t_iu$ -- a contradiction {to the properties of a geometric embedding}.
Consequently, $(\gamma_a\cup\gamma_b)\cap \loop=\{x,x'\}$.} In particular, we obtain \emph{Property~2}: $\gamma_a$ intersects each \rafter of $\Tint_i$ exactly once. This implies that  $a_i$ intersects {each rafter in $D$ exactly once}
and thus proves \ref{item:i}.
\medskip

Together, Properties 1 and 2 imply that $\gamma_a$ visits any section of $\Tint _i$ exactly once. Consequently, the part of $a_i$ covered {by the roof faces of} any section forms an interval on $a_i$. This proves~\ref{item:iii}.
\medskip

It remains to prove  \ref{item:ii}.
By Property 2, $\gamma_a$ intersects each \rafter exactly once. We denote the intersection points on the curve~$\gamma_a$ with the rafters by $p'_1,\dots,p'_{k}$ in the order they appear on $\gamma_a$,  {i.e., $p'_j$ denotes the intersection of $\gamma_a$ with the $j$-th rafter of  $\Tint_i$.}
 Let $p_j$ denote the projection of $p'_j$ onto~$S$, i.e., the intersection point of the $j$-th rafter with $a_i$
 in $D\cup\A$.
We show that $p_1,\ldots,p_k$ appear in this order along $a_i$ in $D\cup \{a_i\}$, i.e., $p_j$ lies before $p_{j+1}$ for all $j$.
To this end, we restrict our attention to the wedge $W_i\subset P_i$ bounded by the rays $uq_1$ and $uq_2$. For simplicity of the presentation, we transform $W_i$ such that (the interesting part of) $e_i$ is on the $x$-axis, $s_i$ left of $t_i$, and that the apex~$u$  at infinity above, see also \cref{fig:orderA}.

    \begin{figure}[htb]
    \centering
    \begin{subfigure}[b]{0.45\textwidth}
    \centering
    \includegraphics[page = 10]{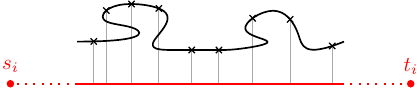}
    \caption{The curve $\gamma_a$ shields $e_i$ from $u$.}
    \label{fig:orderA}
    \end{subfigure}
    \hfil
     \begin{subfigure}[b]{0.25\textwidth}
    \centering
    \includegraphics[page = 11]{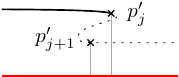}
    \caption{$p'_{j+1}$ lies below $\gamma'$.}
    \label{fig:orderB}
    \end{subfigure}
    \hfil
      \begin{subfigure}[b]{0.25\textwidth}
    \centering
    \includegraphics[page = 12]{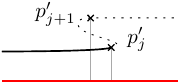}
    \caption{$p'_{j+1}$ lies above $\gamma'$.}
    \label{fig:orderC}
    \end{subfigure}
    \caption{Illustration for the proof of \cref{clm:order}~\ref{item:ii}.}
    \label{fig:order}
    \end{figure}

    Suppose for the purpose of a contradiction that $p_{j+1}$ lies left of $p_j$, see also \cref{fig:orderB,fig:orderC}. 
    Note that the initial subcurve  of $\gamma_a$ until $p'_j$, denoted by $\gamma_a'$, shields the initial part of $e_i$ until $p_j$. 
    If $p'_{j+1}$ lies below $\gamma_a'$  as depicted in \cref{fig:orderB}, then $p'_{j+1}$ is not visible from $u$. 
    However, by construction, all \rafter{}s and thus all points $p'_1, \dots, p'_k$ are visible from the apex $u$. A contradiction. If $p'_{j+1}$ lies above $\gamma_a'$ as depicted in \cref{fig:orderC}, then $p'_{i+1}$ is not visible from $u$  because of the subcurve $\gamma_a-\gamma_a'$, a contradiction. Consequently, $p_j$ is left of $p_{j+1}$ for all $j$.
\end{proof}

In order to define a notion of \emph{left} and \emph{right} of $a_i$ in $D\cup \{a_i\}$, we enhance $D\cup \{a_i\}$ by the projection of $\triangle$. 
The arc $a_i$ partitions the interior of (the projection of) $\triangle$ into to two regions. 
Considering the orientation of $a_i$, it is easy to
determine if a point within one of these two regions lies \emph{left} or \emph{right}.

Every drawing on the sphere can be transformed to a drawing in the plane where any face can be chosen as the outer face. 
In order to relate notions such as clockwise and counter-clockwise with our drawing on the sphere, we fix the face $f_0$ as the outer face {of $D^+$ such that $v_0,v_1,v_2$ is a counter clockwise cycle. For $D$, these notions are then inherited because we consider a drawing $D$ that is an induced subdrawing of $D^+$}.

\begin{claim}\label{clm:LeftRightPaths}
In $D\cup \{a_i\}$, the vertices of the left and right roof paths of $\Tint _i$ are left and right of $a_i$. 
\end{claim}
\begin{proof}

Let $u_1,\dots, u_k$ denote the vertices of the left roof path,  $v_1,\dots, v_k$ the vertices of the central roof path, and  $w_1,\dots, w_k$ the vertices of the right roof path of $\Tint_i$ in $D\cup \{a_i\}$. 
The value of $k$ depends on the number of tunnel sections of $\Tint_i$, namely,  $k=2n-4+i$.
For each $j = 1,\ldots,k$,
it holds that
% $u_i,v_i,w_i$
$u_j,v_j,w_j$
is a \rafter which $a_i$ intersects exactly once by \cref{clm:tunnelRoof}~\ref{item:i}. 
Consequently, if 
% $u_i$
$u_j$
lies left of $a_i$ then 
% $w_i$ 
$w_j$
lies right of $a_i$ and vice versa.

We consider a visible section of $\Tint_i$. To this end, let
% $u_iu_{i+1}$ 
$u_ju_{j+1}$
be a visible edge. 
We show that both 
% $u_i$ 
$u_j$
and 
% $u_{i+1}$ 
$u_{j+1}$
lie left of $a_i$ implying that 
% $w_i$ and $w_{i+1}$. 
{$w_j$ and $w_{j+1}$ lie right of $a_i$.}
Because each vertex of the left or right roof path belongs to some visible section, this property implies the claim.

First consider the case that $u_j$ and $u_{j+1}$ lie right of $a_i$. By \cref{clm:order}~\ref{item:ii}, the intersection $x_j$ of the \rafter $u_j,v_j,w_j$ with $a_i$ lies before the intersection $x_{j+1}$ of the \rafter $u_{j+1},v_{j+1},w_{j+1}$ with $a_i$. 
However,  the cycle $u_{j},u_{j+1},v_{j+1}, w_{j+1}, w_{j}, v_{j}$ is flipped, i.e., it has the outer face $f_0$ to its right, however it is supposed to lie to its left, see \cref{fig:leftRightA,fig:leftRightB}. A contradiction. 
    \begin{figure}[htb]
\centering
\begin{subfigure}[b]{0.3\textwidth}
\centering
\includegraphics[page = 14]{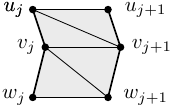}
\caption{}
\label{fig:leftRightA}
\end{subfigure}
\hfil
\begin{subfigure}[b]{0.3\textwidth}
\centering
\includegraphics[page = 15]{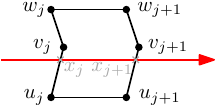}
\caption{}
\label{fig:leftRightB}
\end{subfigure}
\hfil
\begin{subfigure}[b]{0.3\textwidth}
\centering
\includegraphics[page = 16]{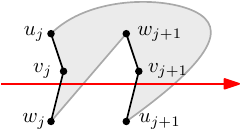}
\caption{}
\label{fig:leftRightC}
\end{subfigure}
\hfil

\caption{Illustration for the proof of \cref{clm:LeftRightPaths}.
}
\label{fig:leftRight}
\end{figure}

Hence it remains to consider the case in which $u_j$ and $u_{j+1}$ lie on different sides.
We consider the case that $u_j$ lies left and $u_{j+1}$ right of $a_i$ 
as illustrated in \cref{fig:leftRightC}. 
By crossing-freeness, exactly one of the edges $u_ju_{j+1}$ and $w_jw_{j+1}$ crosses $a_i$
between $x_j$ and $x_{j+1}$, while the other edge crosses before $x_j$ or after $x_{j+1}$ as depicted in \cref{fig:leftRightC}. 
However, then the roof faces of this section do not cover a consecutive part of $a_i$. A contradiction to  \cref{clm:tunnelRoof}~\ref{item:iii}.
\end{proof}

\medskip

\begin{claim}\label{clm:crossing}
   In $\A$, for all $i,j$ with $i<j$, 
    the arc $a_j$ crosses the arc $a_i$ from left to right.
\end{claim}
\begin{proof}
We restrict our attention to the vertices of the left and right roof paths of $\Tint _i$ and $\Tint _j$. 
Let $\alpha_1,\alpha_2,\alpha_3$ and  $\beta_1,\beta_2,\beta_3$ denote the vertices of two rafters of $\Tint _i$ such that they  are contained in the right and left roof path of $\Tint _j$, respectively. For an illustration, consider \cref{fig:crossingA}.
We call their induced rafters, the $\alpha$-\rafter and the $\beta$-\rafter.
By \cref{clm:LeftRightPaths}, the vertices of a left/right roof path lie on the left/right of 
its corresponding arc.
 \begin{figure}[htb]
\centering
\begin{subfigure}[t]{0.35\textwidth}
\centering
\includegraphics[page = 21]{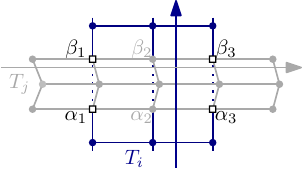}
\caption{The vertices $\alpha_1,\alpha_2, \alpha_3,\beta_1,\beta_2,\beta_3$.}
\label{fig:crossingA}
\end{subfigure}
\hfil
\begin{subfigure}[t]{0.3\textwidth}
\centering
\includegraphics[page = 22]{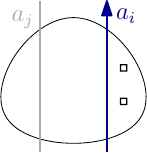}
\caption{The arcs $a_i$ and $a_j$ cross.}
\label{fig:crossingB}
\end{subfigure}
\hfil
\begin{subfigure}[t]{0.3\textwidth}
\centering
\includegraphics[page = 23]{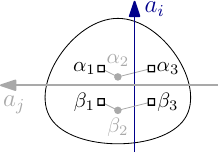}
\caption{Arc $a_j$ cannot cross $a_i$ from right to left.}
\label{fig:crossingC}
\end{subfigure}
\hfil
\caption{Illustration for the proof of \cref{clm:crossing}.}
\label{fig:crossing}
\end{figure}    

Suppose the arcs $a_i$ and $a_j$ do not cross. 
By symmetry, we may assume that $a_j$ lies left of $a_i$, see \cref{fig:crossingB}. 
Recall that the vertices $\alpha_3$ and $\beta_3$ lie to the right side of $a_i$, but on different sides of $a_j$.
A contradiction.

Now, suppose $a_j$ crosses $a_i$ from right to left as in \cref{fig:crossingC}. Then each vertex among $\alpha_1,\alpha_3,\beta_1,\beta_3$ lies in one of the four regions defined by $a_i$ and $a_j$. Moreover $a_j$ separates the $\alpha$-\rafter of $\Tint _i$ containing $\alpha_1,\alpha_2$, and $\alpha_3$ from the $\beta$-\rafter containing $\beta_1$,$\beta_2$, and $\beta_3$. This implies that $a_i$ intersects  the $\beta$-\rafter   before the $\alpha$-\rafter. A contradiction to \cref{clm:order}~\ref{item:ii}.
\end{proof}

\begin{claim}\label{clm:combinatorics}
In $\A$, the order of intersections on each arc $a_i$ is the same as for $\ell_i$ in $L^*$.
\end{claim}
\begin{proof}
We consider any three arcs $a_i,a_j,a_k$ with $i<j<k$.
By \cref{clm:crossing}, the vertices $s_i$,$s_j$,$s_k$,$t_i$,$t_j$, and $t_k$ lie on the correct sides for all three arcs. 
It is therefore the case that
 % Note that
 the order of intersections is correct on one arc of a triple if and only if it is correct on all three arcs.
Suppose for a contradiction that order of intersections on $a_i$ is not correct, i.e., the intersection of $a_j$ and $a_k$ lies on the wrong side of $a_i$. By symmetry, we consider the case that the intersection is left of~$a_i$ while it is supposed to lie right of~$a_i$, see \cref{fig:combinatorics}. 

 \begin{figure}[htb]
\centering
\begin{subfigure}[t]{0.45\textwidth}
\centering
\includegraphics[page = 24]{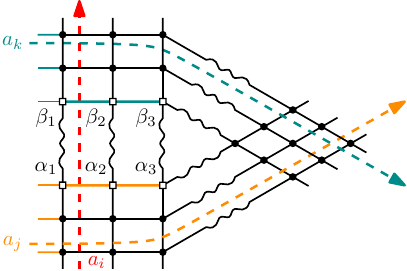}
\caption{The  subgraph of $G$ induced by vertices of the three tunnels $\Tint _i,\Tint _j,\Tint _k$.}
\label{fig:combinatoricsA}
\end{subfigure}\hfil
\begin{subfigure}[t]{0.45\textwidth}
\centering
\includegraphics[page = 25]{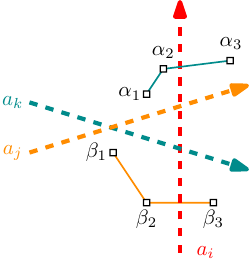}
\caption{A wrong order of intersections implies that~$a_i$ intersects the $\beta$-rafter before the $\alpha$-rafter.}
\label{fig:combinatoricsB}
\end{subfigure}
\caption{Illustration for the proof of \cref{clm:combinatorics}.}
\label{fig:combinatorics}
\end{figure}

{Let $\alpha_1,\alpha_2,\alpha_3$ and  $\beta_1,\beta_2,\beta_3$ denote the vertices of two rafters of $\Tint _i$ such that they  are contained in the left roof path of $\Tint _j$ and in the right roof path of $\Tint _k$, respectively. For an illustration, consider \cref{fig:combinatoricsA}.} Consequently, by \cref{clm:LeftRightPaths},  $\alpha_1,\alpha_2,\alpha_3$ lie left of $a_j$ and $\beta_1,\beta_2,\beta_3$ lie right of $a_k$. This implies that the intersection of the $\beta$-\rafter with $a_i$ lies before the intersection of the $\alpha$-\rafter with~$a_i$. A contradiction to \cref{clm:tunnelRoof}~\ref{item:ii}.
\end{proof}

\begin{claim}
    In $\A$, all intersections lie in {one hemisphere of $S$ bounded} by $a_0$.
\end{claim}
\begin{proof}
Consider the plane $P$ through the arc $a_0$ and apex $u$ and assume without loss of generality that $P$ is horizontal. By construction and \cref{clm:combinatorics}, the first crossing on  each arc $a_i$ is with the arc $a_0$; all other crossings come afterwards. 
Moreover, \cref{clm:crossing} ensures that each arc $a_i$ crosses $a_0$ from left to right. 
Consequently, all arcs start in one hemisphere of $S$ cut by $P$; let us suppose it is the top hemisphere. Then, the intersection of any two arcs $a_i$ and $a_j$ lie on the bottom hemisphere; otherwise the arcs contain both intersections of the supporting great circles.
\end{proof}

Let $P'$ denote a plane obtained by shifting the plane $P$ through $a_0$ and $u$ to the bottom, see also \cref{fig:apexC}.
Then, we 'project' the arc arrangement onto  the plane $P'$ by considering the intersection of $P'$ with the plane supporting the arc $a_i$. Clearly, this intersection yields a line. Hence, we obtain a line arrangement in the plane $P'$ that inherits the combinatorics of $\A$ and thus, proves the stretchability of $L$. 
This finishes the proof of \Cref{thm:TwoThree}.
\end{proof}

%%%%%%%%%%%%%%%%%%%

\paragraph{Fattening the Complex.}
% Let us remark, that 
The constructed 2-complex in the proof of \cref{thm:TwoThree} was not pure. 
Specifically, the special edges are not contained in any triangle. 
We can obtain a pure 2-complex by adding one new vertex to each special edge such that it forms a special triangle. 
Similarly, we can add a private vertex to each triangle to form a pure 3-complex~$C'$. 
Given a geometric embedding of $C$, these new vertices can easily be added close enough to their defining set in $C$. 
Hence, $C$ has a geometric embedding if and only if $C'$ has a geometric embedding in~$\R^3$. 
Together with \cref{thm:TwoThree}, this very small modification shows hardness of \GEM{3}{3}.

 \begin{lemma}\label{lem:threethree}
 The decision problem \GEM{3}{3} is \ER-hard.
 \end{lemma}

%%%%%%%%%%%%%%%%%%%
\newpage
\subsection{Dimension Reduction}
\label{sec:dim-reduction}
%%%%%%%%%%%%%%%%%%%

In order to show hardness for all remaining cases of \Cref{thm:main}, 
we establish the following dimension reduction. 
For dimension reductions in the context of PL embeddings, we refer to \cite{parsa2018links, skopenkov2019short, undecidable2020, parsa2020}.

\begin{lemma}\label{lem:OneUp}
The decision problem \GEM{k}{d} reduces to \GEM{k+1}{d+1}.
\end{lemma}

The idea is to add two apices to a
$k$-complex~$C$ in order to obtain a $(k+1)$-complex $C^+$.
We will then argue that $C$ has a geometric embedding in $\R^d$ if and only if $C^+$ has a geometric embedding in $\R^{d+1}$.
More formally, for a complex $C=(V,E)$  and a disjoint vertex set $U$,  $C*U$ denotes the \textit{join} complex $(V\cup U, E')$ where $E':=\{e\cup u\mid e\in E, u \in U\}$.
The following claim  immediately implies \cref{lem:OneUp}.

\begin{claim}\label{lem:dimensionReduction}
Let $C=(V,E)$ be a complex, $a,b\notin V$ two new vertices, and $C^+:=C*\{a,b\}$ their join complex.
Then $C$ has a geometric embedding in $\R^d$ if and only if $C^+$ has a geometric embedding in $\R^{d+1}$.
\end{claim}
\begin{proof}
Let $\varphi$ be a geometric embedding of $C$ in  $\R^d$. Then, we define for $v\in V\cup \{a,b\}$,
\[\varphi'(v)=
\begin{cases}
(\ \ \varphi(v)\ \ ,\ 0\ )& \text{ if } v\in V,\\
(0,\dots, 0, +1)& \text{ if } v=a,\\
(0,\dots, 0,-1)& \text{ if } v=b.
\end{cases}\]
It is easy to check that $\varphi'$ is a geometric embedding of $C^+$ in $\R^{d+1}$: While $a$ and $b$  are well separated in the last coordinate, all other potential intersections happen in the $d$-dimensional subspace induced by the first $d$ coordinates. Hence $\varphi$ implies the correctness of the geometric embedding.
\smallskip

For the reverse direction, consider a geometric embedding $\varphi$ of $C^+$ in $\R^{d+1}$. 
Let $\varphi_a:=\varphi(a)$ and $\varphi_b:=\varphi(b)$.
Without loss of generality, we assume that $\varphi_a-\varphi_b$ is orthogonal to the first $d$ coordinates, i.e., $\varphi_a-\varphi_b$ is parallel to the $(d+1)$-st coordinate axis.
Let $\overline{\varphi}(C):=\bigcup_{e\in E}\overline{\varphi}(e)$ denote the induced geometric subrepresenation of $C$.
We claim that the function $f:\overline{\varphi}(C)\to \R^d$ obtained by restricting to the first $d$ coordinates is injective. Thus $\varphi':=f\circ \varphi$ yields a representation  of $C$ in $\R^d$.

For the purpose of a contradiction, suppose that $f$ is not injective. Then there exist two distinct points $p =(p_1,\ldots,p_{d+1})$ and $q =(q_1,\ldots,q_{d+1})$ with $p,q\in \overline{\varphi}(C)$ such that $(p_1,\ldots,p_{d}) = (q_1,\ldots,q_{d})$ and $p_{d+1} \neq q_{d+1}$.
Without loss of generality, we may assume that $p_{d+1} > q_{d+1}$.
Consider the plane $P$ spanned by $\varphi_a,\varphi_b,p$. 
Note that $q\in P$. For an illustration, see \cref{fig:monotone}.
\begin{figure}[htb]
\centering
\begin{subfigure}[t]{0.45\textwidth}
\centering
\includegraphics[page =3]{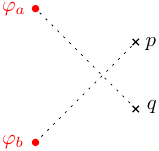}
\caption{}
\label{fig:monotone}
\end{subfigure}
\hfil
\begin{subfigure}[t]{0.45\textwidth}
\centering
\includegraphics[page =4]{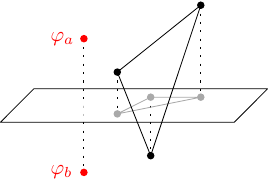}
\caption{}
\label{fig:dimRed}
\end{subfigure}
\caption{Illustration for the proof of \cref{lem:dimensionReduction}. The geometric embedding~$\varphi$ of $C^+$, gives a monotone embedding of $C$, otherwise we can find an intersection.}
\label{fig:dimReduction}
\end{figure}
Let us denote with $e_p \in E$ and $e_q\in E$ any choice of hyperedges such that
$p\in \overline{\varphi}(e_p)$ and $q\in \overline{\varphi}(e_q)$.
Consider the two open segments $\seg^\circ(\varphi_a,q) \in \overline{\varphi}(e_q\cup a)$
and $\seg^\circ(\varphi_b,p) \in \overline{\varphi}(e_p\cup b)$. Clearly, these open segments intersect in a point $x$, as illustrated \cref{fig:monotone}.
Because $\varphi$ is a geometric embedding, it holds that 
\[x\in \overline{\varphi}(e_q\cup a)\cap\overline{\varphi}(e_p\cup b)=\overline{\varphi}(e_q\cap e_p)=\overline{\varphi}(e_q)\cap \overline{\varphi}(e_p).\]
 In particular, this implies that $x\in\overline{\varphi}(e_q)$ and thus that $x\in\seg^\circ(\varphi_a,q)\cap \overline\varphi(e_q)$.
However, because $\overline\varphi(e_q\cup a)$ is a simplex, $\varphi_a$ does not lie in the span$(\overline\varphi(e_q))$ and thus $\seg^\circ(\varphi_a,q)\cap \overline\varphi(e_q)=\emptyset$.
A contradiction. 
\end{proof}

\section{Conclusion}
We established the computational complexity of \GEM{k}{d} for all $d\geq 3$ and $k\in\{d-1,d\}$. 
In particular, we showed that for these values  it is complete for  \ER to distinguish PL embeddable $k$-complexes in $\R^d$ from geometrically embeddable ones.
Arguably, \GEM{2}{3} is the most interesting case.

Investigating the computational complexity for the remaining open entries in \cref{tab:results} remains for future work. 
We strengthen the conjecture of Skopenkov~\cite{skopenkov2020invariants} as follows.
\begin{conjecture*}
The problem \GEM{k}{d} is \ER-complete for all $k,d$ such that $\max\{3,k\}\leq d\leq 2k$.
\end{conjecture*}

\section*{Acknowledgements.}
We thank Arkadiy Skopenkov for his kind and swift help with acquiring literature. 
We thank Martin Tancer for pointing out a mistake in a previous version of this manuscript.
Mikkel Abrahamsen is part of Basic Algorithms Research Copenhagen (BARC), generously supported by the VILLUM Foundation grant 16582.
Linda Kleist is generously supported by a postdoc fellowship of the German Academic Exchange Service (DAAD).
Tillmann Miltzow is generously supported by the Netherlands Organisation for Scientific Research (NWO) under project no. 016.Veni.192.250.

\begin{table}[htp]
\centering
\caption{A glossary for notions used in the proof of \cref{thm:TwoThree}.}
\label{tab:glossary}
\begin{tabular}{c|c}
symbol  &  meaning \\
\hline
$L= \{\ell_1,\ldots, \ell_n\}$ & pseudoline arrangement\\
$n$ & number of pseudolines in $L$, \Cref{fig:pseudolines} \\
$\ell_0$ & additional pseudoline \\
$L^*:=L\cup \{\ell_0\}$ & $L$ together with $\ell_0$ \\
$\triangle$ & the \textit{\htriangle}, all intersections of $L^*$ are contained inside\\
$C = (V,E)$ & the simplicial complex that we construct\\
$\Tall_i$ & tunnel around pseudoline $\ell_i$, \Cref{fig:tunnel}\\
$\Tint_i$ & tunnel $\Tall_i$ without first and last two sections\\
section & part of tunnel, \Cref{fig:section}\\
tunnel bottom & the part of tunnel $\Tall_i$ that lies in the plane  $z = i$ in step 1\\
tunnel roof & the part of tunnel which is not in the bottom, see \cref{fig:tunnelAbove}\\
left/right roof path &  tunnel paths shared by roof and bottom.  \\
central roof path & tunnel path that is disjoint from left/right roof path \\
$e_i = (s_i,t_i)$ & \textit{special edge} of $C$ that is meant to represent $\ell_i$ \\
$u$ & apex (taking the role of an observer high above) \\
$H$ & graph in the third step of the construction, see black graph in \cref{fig:3-connected}\\
$G^+$ & an essentially $3$-connected planar graph induced by the neighborhood of $u$ \\
$f_0$ & outer face of $G^+$ \\
$G $ & subgraph of $G^+$ that is inside the \htriangle \\
$\overline C$ & the complex $C$ consists of two copies of $\overline C$ with the apex identified\\
$\varphi$ & geometric embedding of $C$ \\
tent & \Cref{fig:RedinTunnel} \\
\sloop~\loop & triangle that is shared by two sections\\
\rafter & two roof edges of a \sloop\\
wedge $W_i$ & defined by the apex $u$ and $e_i$, \Cref{fig:aboveBelow} \\
plane $P_i$ & defined by the apex $u$ and $e_i$\\
$p>e_i$ & segment $pu$ does not cross $e_i$  \\
$p<e_i$ & segment $pu$ does cross $e_i$  \\
$S$ & sphere around apex $u$\\
$D$, $D^+$ & projection of $\varphi(G^+)$ onto $S$ yielding a drawing of $G$ and $G^+$\\
$a_i$ & projection of $\varphi(e_i)$ onto $S$ \\
\A  & arc arrangement of all $a_i$~\Cref{fig:apexB}
\end{tabular}
\end{table}

\newpage

\def\bibfont{\small}

\printbibliography
%\bibliographystyle{aomplain}%{plainurl}%{plaindin}
%\bibliography{library/abbrv,library/ETR,library/Simplicial}

\end{document}